\documentclass{article}
\usepackage{amsmath,amssymb,amsfonts,amsthm,graphicx}
\usepackage{tkz-graph}
\usepackage{tikz}
\usepackage{color}
\usepackage{cite}
\usepackage{epsfig}
\usepackage{epstopdf}
\usepackage{footnote}
\usepackage{cite}
\usepackage{caption}
\usepackage{subcaption}
\usepackage{enumerate}
\usepackage{algorithm}
\usepackage{algpseudocode}
\usepackage{times}
\usepackage[top=2cm, bottom=2cm, left=2cm, right=2cm]{geometry}
\usepackage{tkz-graph}
\tikzstyle{vertex}=[circle, draw, inner sep=0pt, minimum size=6pt]

\newtheorem{theorem}{Theorem}[section]
\newtheorem{lemma}[theorem]{Lemma}

\theoremstyle{definition}
\newtheorem{definition}[theorem]{Definition}

\newtheorem{corollary}[theorem]{Corollary}

\theoremstyle{remark}
\newtheorem{remark}[theorem]{\bf Remark}

\DeclareMathOperator{\PN}{PN}
\DeclareMathOperator{\dom}{\mathcal{D}om}

\title{Domination Cover Number of Graphs}
\author{M. Alambardar Meybodi$^{1}$, M.R.~Hooshmandasl$^{2}$, P. Sharifani$^{3}$, A. Shakiba$^{4}$\\
	\footnotesize{$^{1,2,3}$Department of Computer Science, Yazd University, Yazd, Iran.}\\
	\footnotesize{$^{1,2,3,4}$The Laboratory of Quantum Information Processing, Yazd University, Yazd, Iran.}  \\
	\footnotesize{$^4$Department of Computer Science, Vali-e-Asr University of Rafsanjan, Rafsanjan, Iran.}  \\
	\footnotesize{e-mail:$^1$ m.alambardar@stu.yazd.ac.ir, $^2$ hooshmandasl@yazd.ac.ir, $^3$ pouyeh.sharifani@gmail.com}, $^4$ ali.shakiba@vru.ac.ir}

\date{}
\begin{document}

	\maketitle

\begin{abstract}
A set $D \subseteq V$ for the graph $G=(V, E)$ is called a dominating set if any vertex $v\in V\setminus D$ has at least one neighbor in $D$. Fomin et al.\cite{fomin2008combinatorial} gave an algorithm for enumerating all minimal dominating sets with $n$ vertices in $O(1.7159^n)$ time. It is known that the number of minimal dominating sets for interval graphs and trees on $n$ vertices is at most $3^{n/3} \approx 1.4422^n$.  In this paper, we introduce the domination cover number as a new criterion for evaluating the dominating sets in graphs. The domination cover number of a dominating set $D$, denoted by $\mathcal{C}_D(G)$,  is the summation of the degrees of the vertices in $D$. Maximizing or minimizing this parameter among all minimal dominating sets have interesting applications in many real-world problems, such as the art gallery problem. Moreover, we investigate this concept for different graph classes and propose some algorithms for finding the domination cover number in trees, block graphs.
\end{abstract}

\section{Introduction}
The concept of dominating set and its variations are well-studied topics in graph theory. There are thousands of papers on domination in graphs and several well-known surveys and books on this topic such as  \cite{haynes1998domination,haynes1998fundamentals,cockayne1980total}.

The dominating set problem is a classic NP-complete graph problem \cite{michael1979computers}, however, there exist polynomial time algorithms for some graph classes such as trees, interval graphs, and graph with bounded treewidth \cite{cockayne1975linear,chang1998efficient,van2009dynamic}.

Recently, there has been extensive research dealing with enumeration algorithms and combinatorial lower and upper bounds of minimal dominating sets, both on general graphs and on special classes of graphs \cite{golovach2017minimal}.
In \cite{fomin2008combinatorial}, Fomin et al. gave an algorithm with time complexity of $O(1.7159^n)$ to enumerate all of the minimal dominating sets in graphs with $n$ vertices and shown that the number of minimal dominating sets in such a graph is at most $1.7159^n$. They established a lower and an upper bound for the maximum number of minimal dominating sets in graphs. This gap has been narrowed on some well-known graph classes, such as chordal graphs \cite{couturier2013minimal}, trees \cite{krzywkowski2013trees} and co-bipartite graphs \cite{couturier2015number}. Tight bounds have been obtained for some graph classes such as cographs and split graphs \cite{couturier2013minimal,couturier2015number}. Also, it is shown that the interval graphs and trees on $n$ vertices have at most $3^{n/3} \approx 1.4422^n$ minimal dominating sets \cite{golovach2017minimal}.

Let $\mathcal{D}om(G)\subseteq \mathcal{P}(V(G))$ denotes the collection of all dominating sets of $G$ and $\mathcal{D}om_k(G)$ denotes the collection of all dominating sets of $G$ with size $k$ where $k> 0$. It is obvious that $\mathcal{D}om(G) = \cup_{k>0} \mathcal{D}om_k(G)$. The problem of minimum dominating set is to find the smallest $k> 0$ such that $\mathcal{D}om_k(G)\neq \emptyset$ and $\mathcal{D}om_{k'}(G)= \emptyset$ for all ${k'}<k$. Such smallest $k$ is called the minimum domination number of $G$ and is denoted by $\gamma_G$, or simply $\gamma$ when there is no chance of confusions. Any set $D\in \mathcal{D}om_{\gamma}(G)$ is called a  $\gamma-$set of $G$.

In this paper, we propose a new concept, domination cover number, to evaluate  a number of proposed  dominating set like $D$ in a graph $G$ as follows: 

\[\mathcal{C}_D(G)=\sum_{u\in D}|N(u)|=\sum_{u\in D}deg(u),\]
where the index $D$ in  $\mathcal{C}_D(G)$ refer to a special dominating set $D$ of $G$.
We concentrate on investigating maximum ( or minimum) of domination cover numbers. Maximum domination cover number, for $\gamma$-sets, defined as
\[\mathcal{C}^{max}(G)=\max_{D\in \mathcal{D}om_{\gamma}(G)}\mathcal{C}_D(G),\]
and minimum domination cover number, for $\gamma$-sets, defined as
\[\mathcal{C}^{min}(G)=\min_{D\in \mathcal{D}om_{\gamma}(G)}\mathcal{C}_D(G),\]
As an application of maximum domination cover number, consider we want to use minimum number of protective cameras such that maximum overlapping in the protected places occur.

The rest of the paper is organized as follows:
In Section \ref{Notation}, we introduce some essential definitions and notations. In Section \ref{domcover}, we propose the concept of domination cover number and obtain some bounds for it. In Section \ref{coverproduct}, we show some results on the domination cover number in the product of graphs. Next, we investigate some algorithmic approaches in various classes of graphs in Section \ref{coveralgorithm}. Finally, a conclusion is drawn in last section.

\section{Notations and Definitions}\label{Notation}

Let $G = (V, E)$ be a simple graph. The open neighborhood or simply the neighborhood of a vertex $v \in V$ is the set of all vertices adjacent to $v$ and is denoted by $N_G(v)$, i.e. $N_G(v) = \left\{ u \in V| 
\{u,v\} \in E \right\}$. The closed neighborhood of a vertex $v$ is defined as $N_G[v] = N_G(v) \cup \{v \}$. For a subset of vertices $U \subseteq V$, the open and close neighborhoods of $U$ are defined as
$$N_G(U) = \bigcup_{u \in U} N_G(u),$$
and
$$N_G[U] = \bigcup_{u \in U} N_G[u],$$
respectively. Note that when there is no chance of confusion, we may drop the subscript $G$.

Let  $D$ be a subset  of $V$ and $v \in D$.   Then, the private neighborhood of $v$ with respect to $D$, denoted by $\PN[v; D]$, is the set $\PN[v; D] = N[v]\setminus N[D-\{ v \}]$. Each vertex in $\PN[v; D]$ is called a private neighbor of $v$ with respect to $D$.



A set $D_t \subseteq V$ is called a total dominating set of $G$  if every vertex in $V$ is adjacent to at least one vertex in $D_t$. 
The concept of total domination is first introduced by Cockayne, Dawes, and Hedetniemi in 1980 in \cite{cockayne1980total}.
The minimum cardinality among all the total dominating sets of a graph $G$ is denoted by $\gamma_t(G)$ and each total dominating set of size $\gamma_t(G)$ is called a  $\gamma_t$-set.

 A dominating set $D$ is called an efficient dominating set if every vertex $v \in V \setminus D$ is dominated by exactly one vertex in $D$, and every vertex $v \in D$ is not dominated by other vertices in $D$. Such a dominating set $D$ is called an efficient open dominating set, if for every vertex $v \in V$, we have  $|N(v) \cap D| = 1$.
 
 A subset $D \subseteq V$ in a graph $G$ is a total $[1, 2]$-dominating set if, for every
 vertex $v \in V$ , $1 \leq \vert N(v) \cap  D\vert \leq  2$. The minimum cardinality of a total $[1, 2]$-dominating set of $G$ is called the total $[1, 2]$-domination number, denoted by $γ_t[1,2](G)$.

For any connected graph $G$, a vertex $v \in V$ is called a cut-vertex if $G-v$ is not connected, where $G-v$ is exactly the graph $G$ with vertex $v$ and all of its incident edges removed. A maximal connected subgraph of $G$ without any cut-vertices is called a block of $G$.
Moreover, a graph $G$ is called a block graph if its
blocks are complete subgraphs, or equivalently cliques, and the
intersection of any two blocks is either empty or a cut vertex.

The lexicographic product of graphs $G=(V(G),E(G))$ and $H= (V(H),E(H))$ is denoted by $G\circ H$, and is a graph with vertex set $V(G)\times V(H)$, where a vertex $(a,x)$ is adjacent to a vertex $(b,y)$ if either $\{a,b\} \in E(G)$ or $a=b$ and $\{x,y\}\in E(H)$. Let $a\in V(G)$. Then, the induced  subgraph by $\{(a,x)|x\in V(H)\}$ is called the $H$-layer of $G$ with respect to the vertex $a$ and is denoted by $H^a$. It is clear that $H^a$ is isomorphic to $H$.

\section{Domination Cover Number}\label{domcover}
In this section, we give formal definitions for the domination cover number and related problems. 
\begin{definition}[Covering Number of a Set of Vertices]
Let $G=(V,E)$ be a graph and $A\subseteq V$. Then, the covering number of $A$, denoted by $\mathcal{C}_A(G)$, is defined as $\sum\limits_{v\in A}deg_G(v)$.
\end{definition}
\begin{remark}
Note that the covering number of a set of the vertices $A$ can be defined equivalently as $\sum\limits_{v\in A}|N_G(v)|$.
\end{remark}

Since a dominating set $D$ of a graph $G=(V,E)$  is  a subset of vertices, then we can calculate the covering number of $D$, i.e.   
$$\mathcal{C}_D(G)=\sum\limits_{u\in D}deg_G(u).$$
This parameter is called the domination cover number of the dominating set $D$ with respect to $G$.

In this  paper, we tend to solve the following problem which is referred to as the maximum domination covering of $\gamma_G$-sets, denoted by $\mathcal{C}^{max}(G)$. This problem is defined as
\[
	\mathcal{C}^{max}(G)=\max_{D\in \mathcal{D}om_{\gamma}(G)}\mathcal{C}_G(D),
\]
where $\dom_{\gamma}(G)$ is the collection of all $\gamma-$sets.

\begin{remark}
It is worth noting that the $\mathcal{C}^{min}(G)$ can be similarly defined as
\[\mathcal{C}^{min}(G)=\min_{D\in \dom_{\gamma}(G)}\mathcal{C}_G(D).\]
\end{remark}

It is clear that the domination number is unique, however, there may exist several dominating sets of size $\gamma$. The main purpose of introducing the domination cover number is to find a $\gamma$-set with the maximum or minimum covering number. 

\begin{remark}
	Note that the problem of finding maximum (minimum) domination cover number is extensible to other types of dominations as well.
\end{remark}

 \begin{definition}{\textbf{(Total Domination Cover Number)}}
 Let $G$ be a graph and $D_t$ be a total dominating set. The total domination cover number of $G$ with respect to $D_t$ is defined as
 \[\mathcal{C}_{D_t}(G)=\sum_{u\in D_{t}}deg_G(u).\]
 \end{definition}

 Let $H=(W,F)$ be a subgraph of $G=(V,E)$ and  $D\in \mathcal{D}om_{\gamma}(G)$. Then, the domination cover number of $H$ with respect to the set $D$ equals
\[\mathcal{C}_D(H)=\sum_{u\in W}\vert N_G(u)\cap D\vert .\]

 
In the following, we investigate Domination cover number for some graphs.
\begin{theorem}(\cite{chartrand2006introduction}, page 371)
For any path $P_n$ and cycle $C_n$, we have
\[ \gamma(P_n)=\gamma(C_n) =
  \begin{cases}
    k,      & \quad \text{if } n =3k,\\
    k+1,  & \quad \text{otherwise. }\\
  \end{cases}
\]
\end{theorem}

\begin{theorem}Let $D\in \mathcal{D}om_{\gamma}(P_n)$. The domination cover number with respect to  $D$ satisfies:
	\begin{itemize}
		\item if $n=3k$, then  $\mathcal{C}_D(P_n)=2k$, 
		\item if $n=3k+1$, then $2k\leq \mathcal{C}_D(P_n)\leq 2k+2$, 
		\item if $n=3k+2$, then $2k+1\leq \mathcal{C}_D(P_n)\leq 2k+2$. 
	\end{itemize}

\end{theorem}
\begin{proof}
Let $P_n=v_1v_2\cdots v_n$ be a path of length $n$. In the case that $n=3k$, the path $P_n$ has a unique  $\gamma-$set, i.e.  $D=\{v_p\,:\,p=3t+2,  \,\, 0\leq t\leq k-1\}$ is the only dominating set of $P_n$ in this case. It is obvious to see that the degree of each vertex in this set is two, therefore we have $\mathcal{C}_D(P_n)=2k$.
In the case of $n=3k+1$, one of the $\gamma-$sets is  $D=\{v_p : p=3t+1,\,\, 0\leq t\leq k-1\}$. It is easy to see this set is the minimum dominating set with the minimum covering number and it has $k-1$ vertices of degree two and two vertices of degree one. Therefore the lower bound is established. Also, we have another dominating set of size $k+1$ where none of its vertices  has degree one which concludes the upper bound $\mathcal{C}_D(P_n)\leq 2k+2$. 
In the case of $n=3k+2$, the upper and lower bounds are obtained similar to  the previous case.
\end{proof}

There exist many different graph classes such as Barbell graphs and Book graphs that $\mathcal{C}_D(G)=|V(G)|$, e.g. see Figure \ref{Fig:1}.

\begin{figure}[h]\label{Fig:1}
\centering
\begin{subfigure}{.5\textwidth}
  \centering
  \includegraphics[scale=.8]{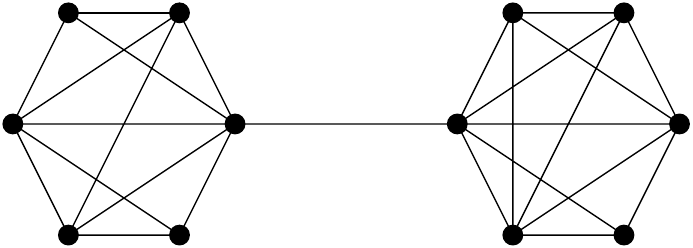}
  \caption{Barbell graph with $n=5$}
\end{subfigure}%
\begin{subfigure}{.5\textwidth}
  \centering
  \includegraphics[scale=.6]{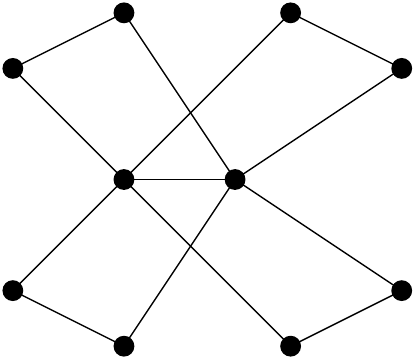}
  \caption{Book graph $B_4$}
\end{subfigure}
\caption{Graphs with a dominating set $D$ such that $\mathcal{C}_D(G)=|V(G)|$}\label{Fig:1}
\end{figure}

Also, there are many graphs with unique domination cover number. For example, let $G$ be a graph with dominating set  $D$ that satisfies the following condition for every vertex $x\in D$
$$\gamma(G-x)\ge \gamma(G).$$
 Then, the set $D$ is the unique dominating set for $G$ \cite{gunther1994graphs}. Therefore, the graph $G$ has a unique domination cover number.

\subsection{Bounds for the Domination Cover Number}
In this section, we provide some bounds for domination cover number of graphs. For convenience, we assume that $G$ is of order $n$.
\begin{lemma}\label{lower_bound}
Let $D\in \mathcal{D}om_{\gamma}(G)$, then we have
\[\mathcal{C}_D(G) \geq n-\gamma(G).\]
\end{lemma}
\begin{proof}
It is clear that if $D$ is an efficient dominating set, then for every vertex $v\in V\setminus D$, we have $\vert N(v)\cap D\vert=1$ and  for  $v\in  D$, we have $\vert N(v)\cap D\vert=0$, so this bound is concluded. In the case that $D$ is not an efficient dominating set, we have two cases to consider. In the first case, there exists at least one vertex $u\in V\setminus D$ such that $\vert N(u)\cap D \vert \geq 2$. So,
\[\mathcal{C}_D(G) =\sum_{v\in V\setminus D} \vert N(v)\cap D \vert \geq 2+ \sum_{v\in V\setminus {(D\cup \{u \})}} 1\geq n+1-\vert D \vert.\]
In the second case, there exists at least one vertex $u\in  D$ such that $\vert N(u)\cap D \vert \geq 1$. Therefor as in the first case, we have $\mathcal{C}_D(G)\geq n+1-\vert D\vert$.
\end{proof}

\begin{corollary}
If $G$ has an efficient dominating set, then 
\[\mathcal{C}^{min}(G)=n -\gamma(G).\]
\end{corollary}

\begin{theorem}\label{bounds}
Let $D\in \mathcal{D}om_{\gamma}(G)$, then the domination cover number of $G$ is bounded as
\begin{equation}
\left \lceil \frac{n}{2} \right \rceil \leq \mathcal{C}_D(G)\leq \left (\left \lceil \frac{n}{2} \right \rceil\right )^2,
\end{equation}
and these bounds are sharp, in the sense that there exist graphs satisfying the equalities for infinitely many values of $n$.
\end{theorem}

\begin {proof}
The lower bound is trivial, by the fact that $\gamma(G)\leq \left \lfloor \frac{n}{2} \right \rfloor$ and Lemma \ref{lower_bound}.
Since $D$ is minimum dominating set, then every vertex $v\in D$ has at least one private neighbor, otherwise if $v$ has not a private neighbor,  we can remove $v$ from $D$ and find a smaller dominating set which is a contradiction. If $v\in D$ has more than one private neighbor, e.g. $u,w\in PN[v;D]$, the vertex $w$ does not change the value of $n-\mathcal{C}_D(G)$. So, we assume every vertex in $D$ has exactly one private neighbor.
With this assumption, the maximum domination cover can be obtained when the vertices in $D$ construct a complete graph and every $v\in D$ also have a private neighbor. Therefor
\[\mathcal{C}_D(G)=\binom{\gamma(G)}{2}+\gamma(G).\]
Since $\gamma(G)\leq \left \lfloor \frac{n}{2} \right \rfloor$, so the upper bound is established.

Next, we prove that these bounds are sharp. Let $H$ be a graph with $n=2p$ vertices which is constructed from base $K_p$ by adding an extra vertex adjacent to each vertex in $K_p$. This graph is called a corona graph.
It is clear that $\gamma(H)=p$, so if the dominating set contains one-degree vertices, then $\mathcal{C}_D(H)=\left \lceil \frac{n}{2} \right \rceil=p$ and if the dominating set contains vertices of $K_p$, then $\mathcal{C}_D(H)=\left \lceil \frac{n}{2} \right \rceil ^2=p^2$.
\end{proof}

A graph $G$ is called $P_4$-free, if $G$ does not contain an induced subgraph $P_4$.  In the next theorem, we establish bounds for domination cover number of  $P_4$-free graphs. These graphs are also known as cographs.

\begin{theorem}\label{lower}
Let $G$ be a $P_4-$free graph with n vertices and  $D\in \mathcal{D}om_{\gamma}(G)$. Then, we have
\[n-\gamma(G)\leq \mathcal{C}_D(G)\leq 2n-\gamma(G).\]
\end{theorem}
\begin{proof}
The lower bound is resulted by Lemma \ref{lower_bound}. In \cite{chellali20131,lv2015total}, it is shown that in $P_4-$free graphs we have $\gamma(G)=\gamma_t[1,2]$, i.e. domination number is equal to total [1,2] domination number of graph. So every vertex is dominated at most twice. 
But according the definition of total $[1,2]$- domination every $v\in D$ can be connected to at most one other vertex in $D$, So we reduce the size dominating set from $2n$.

\end{proof}

\section{Domination Cover Number in Graph Products}\label{coverproduct}
In this section, we investigate domination cover number for the lexicographic product of graphs $G$ and $H$.

\begin{theorem}\label{product-lex}
	The minimum and maximum domination cover number of lexicographic product of $G$ and $H$ are as follows:

		\begin{equation*}
		\mathcal{C}^{min}(G\circ H) =\left\{\begin{array}{ll}
		\mathcal{C}^{min}(G)\times|V(H)|+\gamma(G)(|V(H)|-1),      &   \gamma(H)=1,\\
		\mathcal{C}^{min}_{D_t}(G)\times|V(H)|+\gamma_t(G)\delta(H),  & \gamma(H)>2 \text { or }   (\gamma_t(G)\leq 2\gamma(G)\; \text{and}\; \gamma(H)=2)\\
		\min \left( \alpha^{min}, \beta^{min} \right), &  \text{otherwise,}
		\end{array}\right.
		\end{equation*}
		and 
		
		\begin{equation*}
		\mathcal{C}^{max}(G\circ H) =\left\{\begin{array}{ll}
		\mathcal{C}^{max}(G)\times|V(H)|+\gamma(G)(|V(H)|-1),      &   \gamma(H)=1,\\
		\mathcal{C}^{max}_{D_t}(G)\times|V(H)|+\gamma_t(G)\Delta(H),  & \gamma(H)>2 \text { or }   (\gamma_t(G)\leq 2\gamma(G)\; \text{and}\; \gamma(H)=2)\\
		\max \left (\alpha^{max}, \beta^{max} \right), &  \text{otherwise,}
		\end{array}\right.
		\end{equation*}
		where $\alpha^{min}=\mathcal{C}^{min}_{D_t}(G)\times|V(H)|+\gamma_t(G)\delta(H)$,   $\beta^{min}=2\mathcal{C}^{min}(G)\times|V(H)|+\mathcal{C}^{min}(H)$, $\alpha^{max}=\mathcal{C}^{max}_{D_t}(G)\times|V(H)|+\gamma_t(G) \Delta(H)$ and $\beta^{max}=2\mathcal{C}^{max}(G)\times|V(H)|+\mathcal{C}^{max}(H)$.

\end{theorem}

To prove Theorem \ref{product-lex}, we need the following lemmas.

\begin{lemma}\label{thm1}
Let $S$ be a  dominating set for $G$ with minimum cardinality and $h\in V(H)$ dominates all the vertices in $H$. Then, the set $D=S\times \{h\}$ is a dominating set for $G\circ H$.
\end{lemma}
\begin{proof}
It is sufficient to show the following statements:
\begin{enumerate}
\item The set $S\times \{h\}$ dominates all the vertices in $G\circ H$.
\item There exist no dominating set for $G\circ H$ of cardinality less than $D$.
\end{enumerate} 
Since every vertex in $V(G)\setminus S$ is dominated by a vertex in $S$, then it is easy to see that every vertex in $(w,x)\in N(S\times \{h\})$ is dominated by a vertex of $(w',h)$ such that $\{w,w'\}\in E(G)$. In other words every vertex $(y,z)\in V(G\circ H)$ such that $y \in S$ is dominated by $(y,h)$.

Next, we show that $D$ has minimum cardinality. Suppose $D'$ be a $\gamma$-set for $G\circ H$ and let $S'=\{x:(x,y)\in D'\}$. It is easy to see that $S'$ is a dominating set for $G$ and $\vert S \vert \leq \vert S' \vert$. Therefore, $D$ has the minimum cardinality.
\end{proof}

\begin{lemma}\label{thm2}
Let $S$ be a $\gamma_t$-set for $G$ and $\gamma(H)\geq 2$. Then, for each vertex $u\in V(H)$ the set $D=S\times\{h\}$ is a $\gamma$-set for $G\circ H$.
\end{lemma}
\begin{proof}
We will show that $D$ is a dominating set for $G\circ H$ with the minimum size $\vert S \vert$ among all dominating sets.
Since $S$ is a total dominating set for $G$, then for every $g'\in V(G)$, there exists a vertex $g\in S$ such that $\{g,g'\}\in E(G)$. Therefore,  each vertex $(g',h')\in G\circ H$ is adjacent by a dominating vertex $(g,h)\in G\circ H$ where $h, h' \in V(H)$.

Assume that $D'$ is a dominating set for $G\circ H$ of cardinality less than $|S|$ and let $S'=\{x:(x,y)\in D'\}$ and $S'$ is not a total dominating set. So, there exists $g\in S'$ such that none of its neighbors are in $S'$. It means that there is a vertex $(g,h)\in D$ such that there is no vertex  $(g',h')\in D$ where $\{g,g'\} \in E(G)$. 
Therefore, the vertex $(g,h)$ can be dominated by the vertex $(g,h'')\in D'$ where $\{h,h''\}\in E(H)$.
 For every vertex in $S'$ which is not total dominated, there exist at least two vertices in $D'$. Now, it is enough to select one of  the vertices in $\{(g,v_i):v_i\in V(H)\}$ and instead of other vertices, we select $(g',h)$ where $\{g,g'\}\in E(G)$.
\end{proof}

\begin{lemma}\label{Total product}
Let $S$ be a $\gamma$-set for $G$, $S'$ be a $\gamma$-set for $H$ with cardinality 2 and $\gamma_t(G)=2\gamma(G)$. Then, $S\times S'$ is a $\gamma-$set for $ G\circ H$.
\end{lemma}
\begin{proof}
By Theorem \ref{thm2},  $\gamma(H)\geq 2$  leads to $\gamma (G\circ H)=\gamma_t(G)$. Thus, we have $|S\times S'|=2|S|=2\gamma(G)=\gamma_t(G)=\gamma (G\circ H)$.
 According to the proof of Theorems \ref{thm1} and \ref{thm2}, the set $S\times S'$ is a $\gamma$-set for $G\circ H$.
\end{proof}
\begin{lemma}
Let $S$ be a $\gamma$-set for $G$ such that $\gamma_t(G)=2\gamma(G)$. Then, the set $S$ is an independent dominating set.
\end{lemma}
\begin{proof}
We prove this lemma by contradiction. Let there exist at least two vertices $g,g' \in S$ such that $\{g,g'\}\in E(G)$. For constructing total dominating set $S'$, it is enough to put vertices $g$ and $g'$ in the set $S'$ and for the rest of the vertices in $S$ like $v$, we put $v$ and one of its neighbor in $S'$. As a result, we obtain a total dominating set for $G$ such that it has cardinality of at most $2(\gamma(G)-2))+2=2\gamma(G)-2$ which is a contradiction.
\end{proof}
Let $D\subseteq V(G\circ H)$, for every $(v,u)\in D$ we define $d_G(v)$ and $d_H(u)$ be the degrees of $v\in V(G)$ in $G$ and $u\in V(H)$ in $H$, respectively.

\begin{lemma}\label{degreeseq}
Let $D$ be a $\gamma$-set for $G\circ H$, $\gamma(H)=1$ and $S=\{x:(x,y)\in D\}$. Then the domination cover number  constructed by $D$ is
 \[\mathcal{C}_D(G\circ H)=\mathcal{C}_S(G)+\gamma(G)\times (|V(H)|-1).\]
\end{lemma}
\begin{proof}
Let the set $\bigcup\limits_{(v,u)\in D}\{d_G(v)\times \vert V(H)\vert +\vert V(H) \vert -1\}$ be the degrees for vertices of the dominating set $D$.
So,
\begin{align*}
\mathcal{C}_D(G\circ H)&=\sum_{(v,u)\in D}(d_G(v)\times|V(H)|+|V(H)|-1) \\
& =\sum_{(v,u)\in D}(d_G(v)\times|V(H)|)+\gamma(G)\times(|V(H)|-1) \\
&=\mathcal{C}_S(G)\times|V(H)|+\gamma(G)\times(|V(H)|-1). 
\end{align*}
\end{proof}
\begin{lemma}\label{Lemma3.6}
Let $D$ be a $\gamma$-set for $G\circ H$ and either $\gamma(H)>2$ or $\gamma(H)=2 \text{ and } \gamma_t(G)\leq 2\gamma(G)$. Then,  \[\mathcal{C}_D(G\circ H)=\mathcal{C}_S(G)\times|V(H)|+\sum_{(v,u)\in D}d_H(u).\]
\end{lemma}
\begin{proof}
According to Lemmas \ref{Total product} and \ref{degreeseq}, the set $S$ is a $\gamma_t$-set for $G$. The set 
$\bigcup\limits_{(v,u)\in D} \{d_G(v) \times \vert V(H)\vert +d_H(u)\}$ is degree of vertices in $D$. So we have 
\begin{align*}
\mathcal{C}_D(G\circ H) &=\sum_{(v,u)\in D} (d_G(v) \times \vert V(H)\vert +d_H(u))\\ &=|V(H)|\sum_{(v,u)\in D}d_G(v)+\sum_{(v,u)\in D}d_H(u)\\ &=\mathcal{C}_S(G)\times|V(H)|+\sum_{(v,u)\in D}d_H(u).
\end{align*}
\end{proof}
\begin{lemma}
Let $D$ be a $\gamma$-set for $G\circ H$, $\gamma(H)=2$ and $\gamma_t(G)=2\gamma(G)$. Then,
  \[\mathcal{C}_D(G\circ H)=2\mathcal{C}_S(G)\times|V(H)|+\sum_{(v,u)\in D}d_H(u),\]
  where $S=\{x:(x,y)\in D\}$.
\end{lemma}

\begin{proof}
If $S$ be a $\gamma_t$-set, the proof is similar to the proof of Lemma \ref{Lemma3.6}. otherwise, we  consider the case where the set $S$ is a $\gamma$-set. The set 
$\bigcup\limits_{(v,u)\in D} \{d_G(v) \times \vert V(H)\vert +d_H(u)\}$ contains the degrees of vertices in $D$. So 
\[\mathcal{C}_D(G\circ H)=|V(H)|\sum_{(v,u)\in D}d_G(v)+\sum_{(v,u)\in D}d_H(u)=2\mathcal{C}_S(G)\times|V(H)|+\sum_{(v,u)\in D}d_H(u).\]
\end{proof}
Now, we have all of requirement to do the proof of Theorem \ref{product-lex}.

\begin{proof}
We have just proved that $\mathcal{C}^{min}(G\circ H)$. The proof for $\mathcal{C}^{max}(G\circ H)$ is similar. According to the definition of domination cover number, we have
\[\mathcal{C}^{min}(G\circ H)=\min\{ \mathcal{C}_{D}(G\circ H)\,|\, D \in \mathcal{D}om_\gamma (G\circ H)\}\]
There are three cases to consider:

{\bf Case 1: $\gamma(H)=1$ }\\
In this case, we have:
\begin{align*}
\mathcal{C}^{min}(G\circ H) & =\min\{ \mathcal{C}_{S}(G)\times|V(H)|+\gamma(G)\times(|V(H)|-1): S\in \mathcal{D}om_\gamma(G)\} \\
& =\min\{\mathcal{C}_{S}(G)|V(H)|:S\in \mathcal{D}om_\gamma(G) \}+\gamma(G)\times(|V(H)|-1) \\
&=\mathcal{C}^{min}(G)\times|V(H)|+\gamma(G)\times(|V(H)|-1).
\end{align*}
{\bf Case 2: Either $\gamma(H)>2\,\,\text{or}\,\, \gamma_t(G)\leq 2\gamma(G),\gamma(H)=2$ }
\begin{align*}
In this case, we have
\mathcal{C}^{min}(G\circ H) & =\min\{ \mathcal{C}_{S}(G)  \times|V(H)|+\sum_{(v,u)\in D} d_H(u): S\in \mathcal{D}om_{\gamma_t}(G)\} \\
& =\min\{\mathcal{C}_{S}(G)|V(H)|:S\in \mathcal{D}om_{\gamma_t}(G) \}+\min\{\sum_{(v,u)\in D} d_H(u)\}\\
&=\mathcal{C}^{min}(G)\times|V(H)|+\gamma_t(G)\delta(H). 
\end{align*}
{\bf Case 3: $\gamma(H)=2\,\,\text{and}\,\, \gamma_t(G)= 2\gamma(G)$ }

In this case, we have two types of dominating sets for $G\circ H :$

\begin{enumerate}
\item $S$ is a $\gamma_t-$set for $G$. Then, $D=S\times\{h\}$ where $h\in V(H)$ is a $\gamma-$set for $G\circ H$ which is similar to case 2.
\item
$S$ is a $\gamma-$ set for $G$ and $S'=\{h_1,h_2\}$ is a $\gamma-$set for $H$. In this case, $S\times S'$ is a $\gamma-$set for $G\circ H$ and domination cover number of this set is
\end{enumerate}
\begin{align*}
\mathcal{C}^{min}(G\circ H) & =\min_{S\in \mathcal{D}om(G), S'\in \mathcal{D}om(H)} \{2\mathcal{C}_{S}(G)  \times|V(H)|+\mathcal{C}_{S'}(H)\}\\
& =2\min_{S\in \mathcal{D}om(G)}  \{\mathcal{C}_{S}(G)|V(H)|\}+\min_{S'\in \mathcal{D}om(H)}\{ \mathcal{C}_{S'}(H)\}\\
&=2\mathcal{C}^{min}(G)\times|V(H)|+\mathcal{C}^{min}(H). 
\end{align*}
 In this case, $\mathcal{C}^{min}(G\circ H)$ is the minimum value of the above types.
\end{proof}

\section{Finding Domination Cover Number in Some Graphs}\label{coveralgorithm}
In this section, we use the dynamic programing approach to find the domination cover number of certain classes of graphs. 
\subsection{Domination Cover Number in Trees}\label{Tree}
\subsubsection{Definitions}\label{def-tree}
We first choose an arbitrary vertex of $T$ and consider it as a root. So, from now, we can think of $T$ as a rooted tree. For each vertex $v$ of $T$, we define the following notation :
\begin{itemize}
\item The set $ch(v)$ consists of all children of $v$.
\item $T_v$ denotes the subtree of $T$ rooted at $v$.

\item $m^+ [v]$ is the size of the smallest dominating set of $T_v$ which contains $v$ . This is well-defined because the set of all vertices of $T_v$ is a set of $T_v$ that contains $v$.
\item $m^- [v]$ is the size of the smallest dominating set of $T_v$ which does \emph{not} contain $v$. If no such set exists, we define $m^-[v]$ to be $\infty$.
\item $Max^+[v]$ ($Min^+[v]$) is the size of the maximum (minimum) domination cover number of $T_v$ which contains the vertex $v$ . This is well-defined because the set of all vertices of $T_v$ is a dominating set for $T_v$ that contains $v$.  We first initialize $Max^+[v]$ to zero.

\item $Max^-[v]$ ($Min^-[v]$) is the size of the maximum (minimum) domination cover number of $T_v$ that does not contain $v$. This is well-defined because the set of all vertices of $T_v$ is a dominating set for $T_v$ that does not contains $v$. The initialization value of this parameter is  zero.
\end{itemize}

We are going to devise a linear time algorithm based on a bottom-up dynamic programming technique to calculate the above notation for all of the vertices. To calculate these values for each vertex, we assume that similar values have  been already computed for all of the descendants of the current vertex. This is a valid assumption since we process the vertices of $T$ in post-order.
For convenience, we use black and white colorings for the vertices to denote whether they are in the domination or not, respectively. 

\subsection{Values for leaves}
A vertex $v$ is called a leaf if it has no children.
If $v$ is a leaf of $T$, then $T_v$ consists only of $v$ and therefore, we set $m^+[v]=1$ and $m^-[v]=\infty$, since there exists no dominating set for $T_v$ that does not contain $v$.

\subsection{Calculating $m^-[v]$ when $v$ is not a leaf} \label{mminusv}
Since we are concerned with $m^- [v]$ when $v$ is not a leaf, we have the assumption that $v$ is white. Therefore, at least one of its children must be black. Moreover, $v$ has no effect on how one colors the subtrees $T_u$ for $u \in ch(v)$. We consider all of the possible valid bi-colorings and choose one of the yielding  numbers of black vertices with at least value of $T_v$ as follows:

At least one of the children of $v$ must be in the dominating set.
 So, we construct the set $A$ which consists of all children of $v$ like $w$, which satisfy  $m^+[w]-m^-[w]\leq 0$ and  $B$ is assumed as the set of all of children $v$ satisfying  $m^+[w]-m^-[w]= 0$. Now, we solve the following equation
\begin{equation}
m^-[v]= \sum_{u\in A}m^+[u]+\sum_{u\in ch(v)\setminus A}m^-[u].
\end{equation}

Note that in the above equation, at least one of the  children of $v$ must be chosen. If A is empty, then we select  a vertex from B which has the maximum value of $Max^+$. Moreover, if $B$ is empty, we select the vertex $w$ in $ch(v)$ such that, $m^-[w]-m^+[w]$ is minimum.

The value of $Max^-[v]$ is calculated as 
 \begin{equation}
 Max^-[v]=\sum_{u\in A}(Max^+[u]+1)+\sum_{u\in ch(v)\setminus(A\cup B)}Max^-[u]+\sum_{u\in B}max\{Max^+[u],Max^-[u]\},
 \end{equation}
 
 and similarly, the value of $Min^-[v]$ is calculated as
  \begin{equation}
  Min^-[v]=\sum_{u\in A}(Min^+[u]+1)+\sum_{u\in ch(v)\setminus(A\cup B)}Min^-[u]+\sum_{u\in B}min\{Min^+[u],Min^-[u]\}.
  \end{equation}

\subsection{Calculating $m^+[v]$ when $v$ is not a leaf} \label{mpluscalc}
In this case, the vertex $v$ is black and all of the children of $v$ have a black neighbor and there is no restriction about the color of its children, i.e. each child of $V$ can be either white or black. 
 The value of $m^+[v]$ is calculated as 
\begin{equation}
m^+[v]= 1+\sum_{u\in ch(v)}\min (m^+[u],m^-[u]),
\end{equation}
and to find the value of $Max^+[v]$, we consider all of the selected children of  $v$ in the above equation as the set $A$, i.e.
 \begin{equation}
 Max^+[v]=\sum_{u\in A}(Max^+[u]+2)+\sum_{u\in ch(v)\setminus A}(Max^-[u]+1).
 \end{equation}
 Similarly, for $Min^+[v]$ we have
  \begin{equation}
  Min^+[v]=\sum_{u\in A}(Min^+[u]+2)+\sum_{u\in ch(v)\setminus A}(Min^-[u]+1).
  \end{equation}
The time complexity of calculating all of these equations is $O(deg(v))$.

\begin{remark}
With little changes in the above algorithm, we can find a dominating set for a subtree of $T$ with maximum or minimum domination cover number.
\end{remark}

\subsection{Block graphs}
Our algorithm works on a tree-like decomposition structure, named refined cut-tree of a block graph \cite{aho1974design}. Let $G$ be a block graph with $h$ blocks $B=\{BK_1, \dots , BK_h\}$ and the set of  cut-vertices  $C=\{v_1,v_2,\dots, v_p\}$. The cut-tree of $G$, denoted by $T^B(V^B,E^B)$ 
is defined as $V_B = \{BK_1, \dots , BK_h\}\cup\{v_1,\dots ,v_p\}$ and $E_B = \{\{BK_i, v_j\} | v_j \in BK_i , 1\leq i\leq h, 1\leq j \leq p\}$. It is shown in \cite{aho1974design} that the cut-tree of a block graph can be constructed in linear time by the depth-first-search method. For any block
$BK_i$ of $G$, define $B_i = \{v \in BK_i | v\;\; \text{is not a cut-vertex}\}$ where $1\leq i\leq h$. 

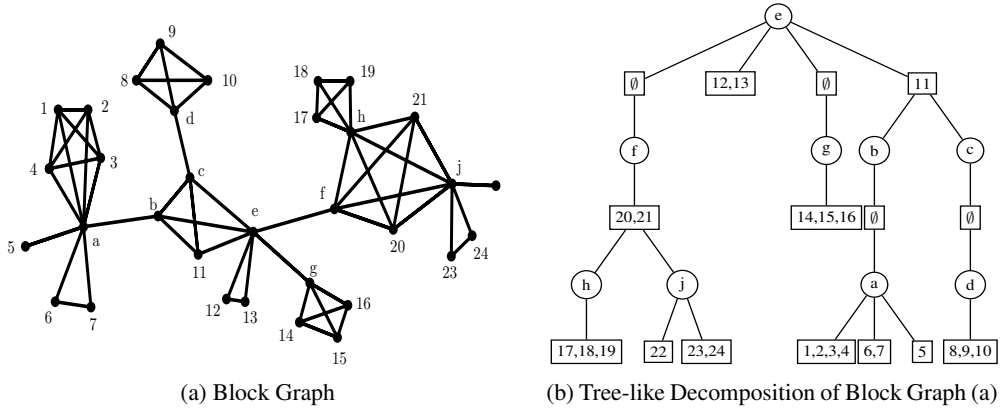
\begin{figure}[h!]

       \begin{subfigure}[b]{0.5\textwidth}
               \centering
               \resizebox{.9\textwidth}{.6\textwidth}{
               \includegraphics{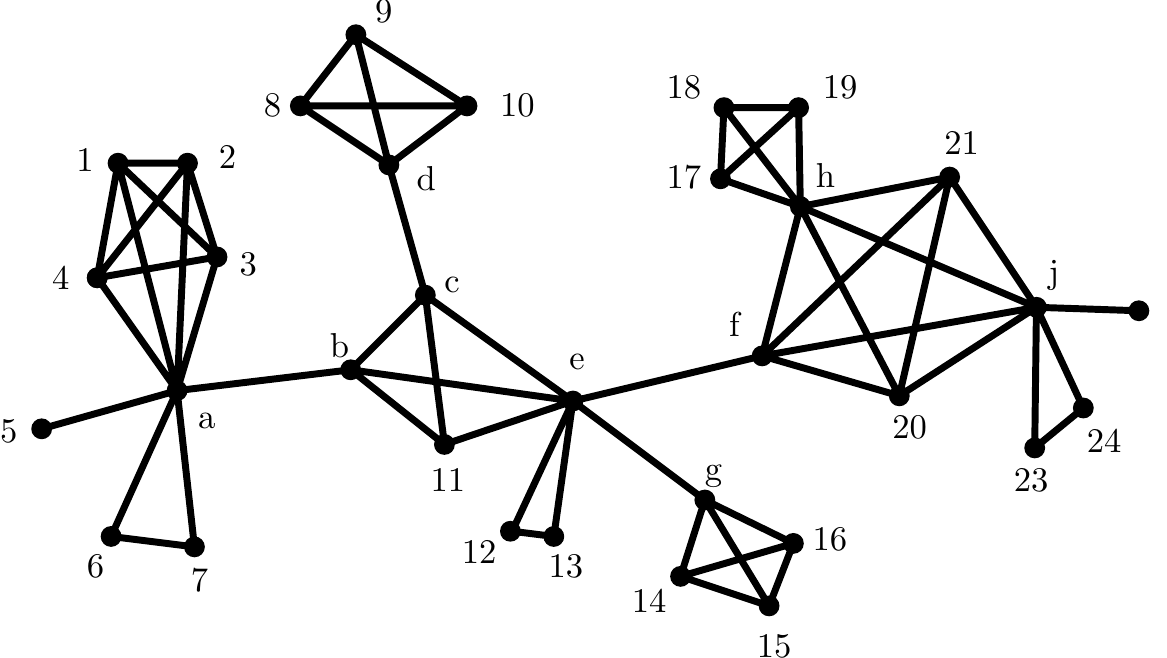}
               }
               \caption{Block Graph}                
       \end{subfigure}  
       \begin{subfigure}[b]{0.4\textwidth}
                \centering
                \resizebox{\textwidth}{.8\textwidth}{
                    \begin{tikzpicture}[thick,
                                every node/.style = {shape=rectangle, draw, align=center,
                                  top color=white}]]
                                \node[shape=circle] {e}
                                [sibling distance=2cm]
                                  child { node {$\emptyset$}                                  	
                                   child { node [shape=circle] {f}
                                   	 child { node {20,21}                                     	
                                    child {[sibling distance=3cm]
                                     node[shape=circle] {h}
                                       child { node {17,18,19} }
                                      }
                                    child {[sibling distance=1cm] node[shape=circle]{j}                                      
                                        child { node {22} }    
                                        child { node {23,24} }
                                     }                                                                   } 
                                   }}                                   
                                  child { node {12,13} }
                                  child { node {$\emptyset$}
                                    child { node[shape=circle] {g}
                                        child { node {14,15,16} } }
                                    }
                                  child { node {11}
                                    child { node[shape=circle] {b}
                                    child { node {$\emptyset$}
                                     child {[sibling distance=1cm] node[shape=circle] {a}
                                     child { node {1,2,3,4}}                                     
                                     child { node {6,7}}
                                     child { node {5}}                                     
                                     }
                                    }                                  }
                                    child { node[shape=circle] {c}
                                    child { node {$\emptyset$}
                                    child { node[shape=circle] {d}
                                    child { node {8,9,10}}
                                    }
                                    }
                                     }
                                   }               ;
                              \end{tikzpicture}}   
                   \caption{Tree-like Decomposition of Block Graph (a)}                
            \end{subfigure}      
             
\caption{Block graph $G$ and its tree-like decomposition}
\end{figure}

\subsection{Domination Cover Number in Block Graphs}
Let $G$ be a block graph and $T$ be its corresponding cut-tree.
We set the following notation:
 \begin{itemize}
 \item $B$ as the set of all block nodes of $T$, 
 \item $C$ as the set of all cut-vertices of $G$,
 \item $m^+_D[v]$ is the size of the smallest dominating set of $T_v$ that  contains node $v$,
 \item $m^-_D[v]$ is the size of the smallest dominating set of $T_v$ that does not contain any vertices of $v$,
 \item $m^-_{\overline{D}}[v]$ is the size of the smallest dominating set of $T_v$ that does not contain any vertices of $v$ but all of vertices in block which is constructed by vertices in $v$ are dominated,
\item  $Max^+[v]$,$Max^-[v]$,$Min^+[v]$ and $Min^-[v]$ are  defined similar to the ones at the beginning of Section \ref{Tree}. 

 \end{itemize}
   
{\bf{Initialization}}

Let $v\in T$ be a leaf. It is clear that $v$ is a block node of $T$  and has some vertices which are not cut-vertices and the degree of all of them is the same. We initialize $m^+[v]=1$ if a non-cut-vertex of $v$ is selected, otherwise $m^+[v]=\infty$,  $m^-_D[v]=\infty,m^-_{\overline{D}}[v]=0$ and also  $Max^+[v]$ equal to the degree of the selected vertex and both  $Max^-_D[v]$ and $Max^-_{\overline{D}}[v]$ to $0$.

{\bf Updating state:}

In the post order traversal of $T$, for each non leaf vertex like $v$, we define the following sets:
\begin{itemize}
\item $S^+=\{u\in ch(v) \;  | \; m^+[u]-m^-[u]< 0\}$,
\item $S^0=\{u\in ch(v)\; |\;  m^+[u]-m^-[u]= 0\}$,
\item $Z=\{u\in S^0 \; |\;  Max^+[u]-Max^-[u]\geq 0\}$.
\end{itemize}
Next, we consider the following cases:

\begin{itemize}
\item[a)] {\bf{$v\in C$:}}
In this case we have two situation to consider:
\begin{itemize}
\item[a.1)] If $v$ is selected, then, all of the children of $v$ are dominated. So, in calculating the domination number we do not consider the leaf children. Thus,
\begin{equation}\label{Block}
m^+[v]=1+ \sum_{u\in ch(v),u \text{ is not leaf}}Min\{m^+[u],m^-_D[u],m^-_{\overline{D}}[u]\}.
\end{equation}

To calculate  the maximum cover number, we first add the degree of $v$ to $Max^+(v)$ and then, for each child $u$, according to  the values of $m^+[u],m^-_D[u],m^-_{\overline{D}}[u]$ as the   minimum value, we add $Max^+(u),Max^-(u),Max^-_{\overline{D}}(u)$ to $Max^+(v)$, respectively. Also, in the case that some of them have the same minimum value, we consider their maximum corresponding  cover number of them.

\item[a.2)]If $v$ is not selected, then there are two cases to consider: 

First, if the set $S^+$ is empty, then all of the children have  been  dominated and we set
\[m^-_D[v]=\sum_{u\in ch(v)}m^-_{\overline{D}}[u],\]
and
\[Max^-_{D}(v)=\sum_{u\in ch(v)} Max^-_{\overline{D}}(u).\]

Second, if $S^+$ is not empty, then we set
\[m^-_{\overline{D}}[v]=\sum_{u\in S^+}m^+[u]+\sum_{u'\in ch(v)\setminus S^+}\min (m^-[u'],m^-_{\overline{D}}[u']),\]
and calculate  $Max^-_{\overline{{D}}}(v)$ by adding $\sum_{u\in S^+} Max^+(u)$ at first and then for each $u'\in ch(v)\setminus S^+$, based on which of $m^-[u']$ or $m^-_{\overline{D}}[u']$ are minimum, we add $Max^-(u')$ or $Max^-_{\overline{D}}[u']$, respectively. If $m^-[u']$ and $ m^-_{\overline{D}}[u']$ are equal, we add the maximum value of $Max^-(u')$ and $Max^-_{\overline{D}}[u']$.

\end{itemize}

\item[b)] If $v$ is a block node and has some non cut-vertices, then we have the following cases to consider:
    \begin{itemize}
    \item[b.1)] If one of the non cut-vertices of $v$, like $x$, is selected, then we set
     \[m^+[v]=1+\sum_{u\in ch(v)}Min (m^+[v],m^-[v]),\]
     and 
     \[Max^+(v)=deg(x)+\sum_{u\in S^+\cup Z} Max^+(u)+\sum_{u\in  ch(v)\setminus (S^+\cup Z)} Max^-(u).\]
    
    \item[b.2)]If $v$ is not selected, i.e. none of the non cut-vertices of $v$ are selected and $v$ has been dominated by at least one of its children, then 
         \begin{itemize}
         \item  In the case that $S^+\cup S^0 \neq \emptyset$, we set
            \[m^-_D[v]=\sum_{u\in S^+}m^+[u]+\sum_{u\in ch(v)\setminus S^-}m^-[u],\]
            and 
            \[Max^-_D(v)=\sum_{u\in S^+\cup Z} Max^+(u)+\sum_{u\in  ch(v)\setminus (S^+\cup Z)} Max^-(u).\]
         \item   If $S^+\cup S^0 = \emptyset$, then without lose of generality,  we can assume that $x$ is the first vertex in the order of $m^+[u']-m^-[u']$ where $u'\in ch(v)$ and if there exist different choices for $x$, we select the vertex with maximum value of $Max^+(x)-Max^-(x)$. Therefore, we set
          \[M^-_D[v]=m^+[x]+\sum_{u\in ch(v)\setminus \{x\}} m^-[u],\]
          and 
          \[Max^-_D(v)= Max^+(x)+\sum_{u\in ch(v)\setminus \{x\}} Max^-(u).\]
       
       \item[b.3)] If $v$ is not selected and is not dominated by at least one of its children, then, we set
          \[m^-_{\overline{D}}[v]=\sum_{u\in ch(v)}m^-[u],\]
          and 
          \[Max^-_{\overline{D}}[v]=\sum_{u\in ch(v)}Max^-[u],\]
               
         \end{itemize}
         
         \item[c)] If $v$ is a block node which does not have any non cut-vertex, then we have $m^+[v]=\infty$ and  both values of $m^-_D[v]$ and $m^-_{\overline{D}}[v]$ can be calculated as in the case b.
    \end{itemize}
  
\end{itemize}

\begin{theorem}
The proposed algorithm has linear time complexity.
\end{theorem}
\begin{proof}
We use a dynamic programming approach over the cut-tree of the block. At each step, the algorithm calculates just some variables and it is straightforward.
\end{proof}

\section{Conclusion}\label{sec:conc}
Throughout the paper, we introduce the new concept as named domination cover number in graph. We find bounds for this parameter in general graphs, paths, $p_4$-free graphs and etc.
Also we investigate domination cover number on the lexicography product of two graphs.
In the last section, two linear algorithms for finding maximum or minimum domination cover number for trees and block graphs are proposed.


\end{document}